%% file: main.tex
\documentclass[journal]{IEEEtran}
\usepackage{amsmath,amsfonts,amssymb}
\usepackage{cite}
\usepackage{bm}
\usepackage{acronym}
\usepackage{graphicx}
\usepackage{orcidlink}

\newtheorem{proposition}{Proposition}
\makeatletter
\def\@begintheorem#1#2{\trivlist\item[\hskip\labelsep{\bfseries #1\ #2:}]\itshape}
\def\@opargbegintheorem#1#2#3{\trivlist\item[\hskip\labelsep{\bfseries #1\ #2\ (#3):}]\itshape}
\makeatother
\newtheorem{remark}{Remark}
\newtheorem{lemma}{Lemma}

\newcommand{\kb}{\bar{\mathbf k}}
\newcommand{\ub}{\bar{\mathbf u}}
\newcommand{\sbb}{\bar{\mathbf s}}
\newcommand{\rb}{\bar{\mathbf r}}
\newcommand{\Hc}{\bm{\mathcal H}}
\newcommand{\Fk}{\bm{\mathcal F}}
\newcommand{\I}{\mathbf I_3}
\newcommand{\Zp}{\mathbf Z_{\!\perp}}
\renewcommand{\Re}{\mathfrak{Re}}
\renewcommand{\Im}{\mathfrak{Im}}

\input{acronyms}

\begin{document}

\title{Wavenumber-Domain Virtual Arrays for\\ Holographic Near-Field Localization}

\author{Giovanni Iacovelli$^{\orcidlink{0000-0002-3551-4584}}$,~\IEEEmembership{Member,~IEEE}, Chandan~Kumar~Sheemar$^{\orcidlink{0000-0003-1676-5983}}$, and Symeon Chatzinotas$^{\orcidlink{0000-0001-5122-0001}}$,~\IEEEmembership{Fellow,~IEEE}
\thanks{G. Iacovelli, C.K. Sheemar, and S. Chatzinotas are with the Signal Processing and Communications (SIGCOM) Research Group at Interdisciplinary Centre for Security, Reliability and Trust (SnT), University of Luxembourg, 1855 Luxembourg City, Luxembourg (emails: giovanni.iacovelli@uni.lu, chandankumar.sheemar@uni.lu, symeon.chatzinotas@uni.lu).}}

\maketitle

\begin{abstract}
Monostatic localization of multiple point targets is studied for a
holographic aperture operated through wavenumber-domain modes. A
specular-point condition delimits the validity of the spectral model as a
near-field approximation. A single snapshot observes a projection of
dimension at most the target count times the polarization components, while
invertible coding recovers the full channel and places the decoded data on
the difference lattice of transmit and receive wavenumbers. Rank conditions
settle identifiability, and the Fisher matrix reduces to a covariance over
the lattice, dictating a nested mode selection that attains full-aperture
resolution with only tens of RF chains.
\end{abstract}


\section{Introduction}
\IEEEPARstart{H}{olographic} Multiple-Input Multiple-Output (MIMO) surfaces and
continuous-aperture arrays (CAPAs) synthesize a quasi-continuous current over an
electrically large aperture \cite{liu2025capa}. Targets of
interest then routinely lie in the radiating near field, where wavefront
curvature enables joint angle and range estimation from a single surface.
Sensing with such apertures has been studied mainly through Cram\'er--Rao
bounds (CRBs), most recently for passive monostatic operation with the
unknown reflectivities profiled out of the likelihood and the continuous
source current optimized \cite{jiang2024crb}, including a wavenumber-domain
reformulation \cite{jiang2025wd}. These works assume a single simultaneous
excitation and localize by exhaustive likelihood search.

A parallel, mature literature builds virtual arrays in the
element-position domain. Colocated MIMO radars separate transmit contributions
by orthogonal waveforms to synthesize $NM$ virtual elements from $N{+}M$
physical ones \cite{stoica2007probing}, and sparse geometries such as nested
arrays achieve $\mathcal O(N^2)$ degrees of freedom from $\mathcal O(N)$
sensors through the difference coarray \cite{pal2010nested}. Two obstacles block a naive transfer to holographic near-field sensing:
single-snapshot operation is degenerate in MIMO radar
\cite{hassanien2015single}, and spherical wavefronts break the far-field
phase additivity on which element-domain virtual arrays rest, to the point
that recent tutorials state the construction fails in the near field
\cite{zeng2025tutorial}. Computational microwave imaging probes scenes
with sequences of aperture modes \cite{sleasman2017single}, but with random
patterns and linear image reconstruction rather than designed mode sets
and parametric localization.

\emph{Contribution.} This letter shows that the virtual-array program survives
the near field if it is carried out in the wavenumber domain of the
aperture rather than in element positions. (i)~A specular-point condition
delimits the validity of the spectral model and identifies the admissible
modal support as a near-field cone (Sec.~\ref{sec:spec}). (ii)~A rank-$3L$ collapse lemma quantifies what a single snapshot loses,
and any invertible coding removes it, placing the decoded data on the
wavenumber-domain difference lattice (Sec.~\ref{sec:virtual}).
(iii)~Identifiability of the joint position and reflectivity estimation
is settled by exact rank conditions, and the profiled Fisher matrix
reduces to a covariance over the lattice, from which the nested selection
and the chains-versus-snapshots rule follow, sustaining an acquisition and
refinement pipeline
(Secs.~\ref{sec:crlb}, \ref{sec:est}).

\section{Electromagnetic Sensing Model}\label{sec:model}
We consider a narrowband communication system where a planar holographic surface
$\mathcal S=[-\mathcal S_x/2,\mathcal S_x/2]\times[-\mathcal S_y/2,\mathcal S_y/2]\subset\{z=0\}$, of area
$|\mathcal S|=\mathcal S_x\mathcal S_y$, illuminates $L$ point targets at
$\mathbf u_\ell=[\ub_\ell^{\mathsf T},z_\ell]^{\mathsf T}$, $z_\ell>0$, $\ell=1,\dots,L$, and collects the backscatter in a monostatic fashion\footnote{Echo
separation is idealized: transmit and receive may be distinct co-located
apertures with negligible spacing, the architecture of metasurface imagers,
where separate panels transmit and collect the return
\cite{sleasman2017single}.}. The time convention is
$e^{-\jmath2\pi ft}$, so outgoing waves carry $e^{+\jmath\kappa R}$ with
$\kappa=2\pi/\lambda$.

\subsection{Probing and reception}
As in metasurface-enabled holographic architectures \cite{sleasman2017single}, the
transmitter synthesizes a monochromatic current density
$\mathbf j(\mathbf s)\in\mathbb C^3$ $[\mathrm{A/m^2}]$ on $\mathcal S$,
decomposed over an orthonormal basis
$\{\bm\phi_n\}_{n=1}^N$ as
$\mathbf j=\sum_n q_n\bm\phi_n$ with
$\int_{\mathcal S}\bm\phi_n^{\mathsf H}\bm\phi_{n'}=\delta_{nn'}$ and probing
coefficients $\mathbf q\in\mathbb C^N$. The radiated
field is
$\mathbf e(\mathbf r)=\jmath\kappa Z_0\int_{\mathcal S}\mathbf G(\mathbf
r,\mathbf s)\mathbf j(\mathbf s)\,\mathrm d\mathbf s$ $[\mathrm{V/m}]$, with the
dyadic Green function
\begin{equation}
\mathbf G(\mathbf r,\mathbf s)=\Big(\I+\tfrac{1}{\kappa^2}\nabla_{\mathbf r}
\nabla_{\mathbf r}^{\mathsf T}\Big)\frac{e^{\jmath\kappa\|\mathbf r-\mathbf s\|}}
{4\pi\|\mathbf r-\mathbf s\|}\in\mathbb C^{3\times3}.
\label{eq:dyadic}
\end{equation}
$\mathbf G$ is symmetric in its arguments, and by reciprocity the
return-path kernel is $\mathbf G$ itself. Target $\ell$ scatters with unknown complex reflectivity
$\gamma_\ell$, which absorbs the constants $(\jmath\kappa Z_0)^2$ and the
cross-section. The echo, superposed over all scatterers, is projected on
an orthonormal receive family $\{\bm\psi_m\}_{m=1}^M$, giving
\begin{equation}
y_m=\sum_{\ell=1}^{L}\gamma_\ell\,\bm\beta_{m,\ell}^{\mathsf H}\bm\alpha_\ell+n_m,
\qquad n_m\sim\mathcal{CN}(0,\sigma^2),
\label{eq:ym}
\end{equation}
\begin{equation}
\bm\alpha_\ell=\!\int_{\mathcal S}\!\mathbf G(\mathbf u_\ell,\mathbf s)\mathbf
j(\mathbf s)\,\mathrm d\mathbf s,\;\;
\bm\beta_{m,\ell}^{\mathsf H}=\!\int_{\mathcal S}\!\bm\psi_m^{\mathsf H}(\mathbf r)
\mathbf G(\mathbf u_\ell,\mathbf r)\,\mathrm d\mathbf r,
\label{eq:alphab}
\end{equation}
where $\bm\alpha_\ell\in\mathbb C^3$ is the field impinging on target $\ell$ and
$\bm\beta_{m,\ell}^{\mathsf H}$ its coupling into receive mode $m$.

\subsection{Generalized (coded) probing}
Let the surface transmit $I$ snapshots, exciting in snapshot $i$ the
superposition with coefficients $[\mathbf Q]_{ni}$, where
$\mathbf Q\in\mathbb C^{N\times I}$ is a known coding matrix, the scene is
static over the $I$ snapshots, and each snapshot carries independent receive
noise. Collecting the received mode vectors as columns of
$\mathbf Y_{\mathrm{meas}}\in\mathbb C^{M\times I}$,
\begin{equation}
\mathbf Y_{\mathrm{meas}}=\mathbf H(\mathbf u,\bm\gamma)\,\mathbf Q+\mathbf N,
\qquad
H_{mn}=\sum_{\ell=1}^{L}\gamma_\ell\,\bm\beta_{m,\ell}^{\mathsf H}\,\bm\alpha_\ell^{(n)},
\label{eq:coded}
\end{equation}
where $\bm\alpha_\ell^{(n)}$ is the field at target $\ell$ due to unit excitation of
mode $n$ alone. The round-trip channel $\mathbf H\in\mathbb C^{M\times N}$
collects the per-mode-pair responses.\footnote{Single-snapshot
probing is the special case $I=1$, $\mathbf Q=\mathbf q$, and single-mode
(sequential) probing is $\mathbf Q=\mathbf I_N$.}

\section{Spectral Representation and Its Validity}\label{sec:spec}
Take spatial harmonics with tangential polarizations,
$\bm\phi_n(\mathbf s)=\mathbf p_n e^{\jmath\kb_n\cdot\sbb}/\sqrt{|\mathcal S|}$,
$\bm\psi_m(\mathbf r)=\mathbf w_m e^{\jmath\kb_m\cdot\rb}/\sqrt{|\mathcal S|}$,
with lateral wavenumbers on the aperture lattice
$\tfrac{2\pi}{\mathcal S_x}\mathbb Z\times\tfrac{2\pi}{\mathcal S_y}\mathbb Z$, written
through integer index vectors $\mathbf t_n,\mathbf r_m\in\mathbb Z^2$ as
$[\kb_n]_a=2\pi t_{n,a}/\mathcal S_a$ and $[\kb_m]_a=2\pi r_{m,a}/\mathcal S_a$,
$k_z(\kb)=\sqrt{\kappa^2-\|\kb\|^2}$, $\Im\{k_z\}\ge0$, and
$\mathbf k=[\kb^{\mathsf T},k_z]^{\mathsf T}$. The Fourier transform of the
Green function follows from the Weyl identity: for $z>0$,
\begin{equation}
\Fk(\kb;z)\!=\!\int_{\mathbb R^2}\!\mathbf G(\bar{\bm\xi},z)
e^{-\jmath\kb\cdot\bar{\bm\xi}}\mathrm d\bar{\bm\xi}
=\Big(\I-\tfrac{\mathbf k\mathbf k^{\mathsf T}}{\kappa^2}\Big)
\frac{\jmath e^{\jmath k_zz}}{2k_z},
\label{eq:weyl}
\end{equation}
the projector enforcing transversality of the radiated wave, which is the
structure lost under scalarization \cite{jiang2024crb}. Define the mode responses
$\mathbf a_n(z)\triangleq\Fk(\kb_n;z)\mathbf p_n\in\mathbb C^{3}$ and
$\mathbf b_m^{\mathsf H}(z)\triangleq\mathbf w_m^{\mathsf H}\Fk(-\kb_m;z)
\in\mathbb C^{1\times3}$. Replacing the truncated aperture integrals \eqref{eq:alphab} by integrals
over the whole plane, so that \eqref{eq:weyl} applies, yields the spectral
approximation
\begin{equation}
\bm\alpha_\ell^{(n)}\approx\tfrac{1}{\sqrt{|\mathcal S|}}\,
\mathbf a_n(z_\ell)\,e^{\jmath\kb_n\cdot\ub_\ell},\qquad
\bm\beta_{m,\ell}^{\mathsf H}\approx\tfrac{1}{\sqrt{|\mathcal S|}}\,
\mathbf b_m^{\mathsf H}(z_\ell)\,e^{-\jmath\kb_m\cdot\ub_\ell},
\label{eq:spectral}
\end{equation}
whose accuracy is settled by the following result.

\begin{proposition}\label{prop:valid}
For a propagating harmonic $\kb_n$ and target $\mathbf u_\ell$, the integrand of
$\int_{\mathcal S}\mathbf G(\mathbf u_\ell,\mathbf s)e^{\jmath\kb_n\cdot\sbb}
\mathrm d\mathbf s$ has a single stationary-phase point
$\sbb^{\star}_{n,\ell}=\ub_\ell-z_\ell\kb_n/k_{n,z}$, around which contributions add
coherently within the first Fresnel zone, of radius
$r_F=\sqrt{\lambda z_\ell/\cos^3\theta_n}$ with
$\sin\theta_n=\|\kb_n\|/\kappa$. The infinite-plane step behind
\eqref{eq:spectral} is accurate, to leading order in $(\kappa z_\ell)^{-1/2}$,
if and only if this coherent disc lies inside
$\mathcal S$, i.e., componentwise in modulus,
\begin{equation}
\big|u_{\ell,a}-z_\ell\,k_{n,a}/k_{n,z}\big|\;\le\;\mathcal S_a/2-r_F,
\qquad a\in\{x,y\}.
\label{eq:validity}
\end{equation}
For the receive integrals, whose weight is $e^{-\jmath\kb_m\cdot\rb}$, the sign
of the walk-off flips, and symmetric mode sets are covered by the same
condition.
\end{proposition}
\begin{IEEEproof}
    Please refer to Appendix~\ref{app:prop1}.
\end{IEEEproof}

\begin{remark}\label{rem:near}
The walk-off $z_\ell\kb_n/k_{n,z}$ grows with range and obliquity: the
model is accurate for close targets and paraxial modes and degrades as the
target recedes, opposite to the usual large-aperture intuition. Only modes satisfying \eqref{eq:validity} carry
information about target $\ell$, a cone $\|\kb_n\|/\kappa\lesssim \mathcal S_a/(2z_\ell)$ for
broadside targets with $z_\ell\gg \mathcal S_a$, which truncates the modal budget of every
selection below.
\end{remark}

Within the validity region, with $c_{mn}(z)\triangleq\mathbf b_m^{\mathsf
H}(z)\mathbf a_n(z)/|\mathcal S|$, the channel entries take the exponential
form
\begin{equation}
H_{mn}=\sum_{\ell=1}^{L}\gamma_\ell\,c_{mn}(z_\ell)\,
e^{\jmath(\kb_n-\kb_m)\cdot\ub_\ell}.
\label{eq:hmn}
\end{equation}
Collect the per-target entries into steering vectors
$[\mathbf h_\ell]_{(m,n)}=c_{mn}(z_\ell)e^{\jmath(\kb_n-\kb_m)\cdot\ub_\ell}
\in\mathbb C^{MN}$, so that
$\mathrm{vec}(\mathbf H)=\sum_\ell\gamma_\ell\mathbf h_\ell$. Since $\Fk$ depends
on $z$ only through $e^{\jmath k_zz}$, the product rule gives
$\partial_zc_{mn}=\jmath(k_{m,z}+k_{n,z})c_{mn}$ and the closed-form gradient is given as
\begin{equation}
\frac{\partial[\mathbf h_\ell]_{(m,n)}}{\partial\mathbf u_\ell}
=\jmath\begin{bmatrix}\kb_n-\kb_m\\ k_{m,z}+k_{n,z}\end{bmatrix}
[\mathbf h_\ell]_{(m,n)}.
\label{eq:grad}
\end{equation}

\section{Virtual Arrays in the Wavenumber Domain}\label{sec:virtual}
\subsection{Single-snapshot rank collapse}\label{sec:rank}
\begin{lemma}\label{lem:rank}
The channel response \eqref{eq:hmn} factorizes as
\begin{equation}
\mathbf H=\sum_{\ell=1}^{L}\frac{\gamma_\ell}{|\mathcal S|}\,
\mathrm{diag}\big(e^{-\jmath\kb_m\cdot\ub_\ell}\big)\,
\mathbf B(z_\ell)\mathbf A^{\mathsf T}(z_\ell)\,
\mathrm{diag}\big(e^{\jmath\kb_n\cdot\ub_\ell}\big),
\label{eq:hfact}
\end{equation}
with $\mathbf B(z_\ell)\in\mathbb C^{M\times3}$,
$\mathbf A(z_\ell)\in\mathbb C^{N\times3}$ stacking
$\mathbf b_m^{\mathsf H}(z_\ell)$, $\mathbf a_n^{\mathsf T}(z_\ell)$. Hence
$\mathrm{rank}(\mathbf H)\le 3L$, and a single snapshot
$\mathbf y=\mathbf H\mathbf q$ observes at most a $3L$-dimensional projection
of the transmit design space.
\end{lemma}
\begin{IEEEproof}
Each summand has rank $\le3$, framed by diagonal (full-rank) phase matrices, and ranks add sub-additively.
\end{IEEEproof}

Low rank itself is benign, the small parametrization that makes the $5L$ unknowns recoverable. The obstacle is observational. In one snapshot, target $\ell$ samples
the illumination at a single point, so every transmit position phase
enters the data only through the condensation
$\mathbf g_\ell=\sum_nq_n\mathbf a_n(z_\ell)e^{\jmath\kb_n\cdot\ub_\ell}$, and
$y_m=\sum_\ell\gamma_ie^{-\jmath\kb_m\cdot\ub_\ell}
\mathbf b_m^{\mathsf H}(z_\ell)\mathbf g_\ell+n_m$: the only position phases
with diversity across the data index $m$ are the receive ones. Any matched
filter draws its spatial response from the $M$ receive wavenumbers alone,
and the grating lobes of a sparse receive selection cannot be
suppressed by transmit design, the wavenumber-domain form of the
single-snapshot degeneracy of MIMO radar \cite{hassanien2015single}.
Equivalently, by Lemma~\ref{lem:rank} the map
$\mathbf q\mapsto\mathbf H\mathbf q$ compresses the $N$-dimensional
design space through rank at most $3L$: transmit modes beyond $3L$ excite
nothing observable within the snapshot. The constant $3$ is the per-target
polarization diversity ($2$ for tangential single-polarization families,
$1$ for scalar models).

\subsection{Decoupling and the difference lattice}\label{sec:code}
Coding across snapshots is the remedy.
\begin{lemma}\label{lem:code}
If $I\ge N$ and $\mathrm{rank}(\mathbf Q)=N$, then
$\widehat{\mathbf H}=\mathbf Y_{\mathrm{meas}}\mathbf Q^{+}$ recovers the
channel. If moreover $\mathbf Q\mathbf Q^{\mathsf H}=c\,\mathbf I_N$
(orthogonal equal-norm rows), the decoded noise is white,
$\mathrm{vec}(\widehat{\mathbf H}-\mathbf H)\sim\mathcal{CN}(\mathbf 0,
\tfrac{\sigma^2}{c}\mathbf I_{MN})$. Examples: $\mathbf Q=\mathbf I_N$
(sequential single-mode probing, $c=1$) or a Hadamard/DFT code transmitting all
modes at full power in every snapshot.
\end{lemma}

The transmit sum is thus removable whenever the snapshot excitations are linearly independent, and all estimation operates on the decoded $\mathrm{vec}(\widehat{\mathbf H})\in\mathbb C^{MN}$. For $I=1$ the observable $\mathbf y=\mathbf H\mathbf q$ is worth only $M$ data dimensions. By contrast,
the
decoded channel of Lemma~\ref{lem:code} has $MN$ entries, each carrying its
own position phase on the difference lattice
$\mathcal D=\{\kb_n-\kb_m\}$, i.e., on the integer differences
$\mathbf d=\mathbf t_n-\mathbf r_m$: a virtual array of $NM$ points in
wavenumber space from $N{+}M$ modes \cite{stoica2007probing}. The
transplant survives the near field, where the element-domain construction
fails \cite{zeng2025tutorial}: the spherical wavefront lives in
$c_{mn}(z_\ell)$ and in \eqref{eq:validity}, while the transverse phases
remain exactly additive in $\kb_n-\kb_m$. Vectorizing the decoded data,
\begin{equation}
\mathrm{vec}(\widehat{\mathbf H})=\Hc(\mathbf u)\bm\gamma+\tilde{\mathbf n},
\quad \Hc=[\mathbf h_1,\dots,\mathbf h_L]\in\mathbb C^{MN\times L},
\label{eq:vecmodel}
\end{equation}
with $\tilde{\mathbf n}\sim\mathcal{CN}(\mathbf 0,\tfrac{\sigma^2}{c}\mathbf
I_{MN})$: linear in $\bm\gamma$, so profiled maximum likelihood applies. The
reflectivities admit the closed form
$\widehat{\bm\gamma}(\mathbf u)=(\Hc^{\mathsf H}\Hc)^{-1}\Hc^{\mathsf H}
\mathrm{vec}(\widehat{\mathbf H})$, leaving the concentrated cost
$V(\mathbf u)=\|\mathbf P^{\perp}\mathrm{vec}(\widehat{\mathbf H})\|^2$ with
$\mathbf P^{\perp}\triangleq\mathbf I_{MN}-\Hc(\Hc^{\mathsf H}\Hc)^{-1}
\Hc^{\mathsf H}$.

\section{CRLB, Identifiability, and Mode Selection}\label{sec:crlb}
\subsection{Bound and identifiability}
All quantities refer to the decoded model \eqref{eq:vecmodel}. Collect the
real parameters
$\bm\theta=[\mathbf u^{\mathsf T},\Re\{\bm\gamma\}^{\mathsf T},
\Im\{\bm\gamma\}^{\mathsf T}]^{\mathsf T}\in\mathbb R^{5L}$. The data are
complex Gaussian with mean $\bm\mu(\bm\theta)=\Hc(\mathbf u)\bm\gamma$ and
covariance $\tfrac{\sigma^2}{c}\mathbf I_{MN}$, so
$\mathbf F=\tfrac{2c}{\sigma^2}\Re\{(\partial\bm\mu/\partial\bm\theta)^{\mathsf
H}(\partial\bm\mu/\partial\bm\theta)\}$ with blocks
\begin{equation}
\frac{\partial\bm\mu}{\partial\mathbf u_\ell}
=\gamma_\ell\frac{\partial\mathbf h_\ell}{\partial\mathbf u_\ell},\qquad
\frac{\partial\bm\mu}{\partial\Re\{\gamma_\ell\}}=\mathbf h_\ell,\qquad
\frac{\partial\bm\mu}{\partial\Im\{\gamma_\ell\}}=\jmath\mathbf h_\ell,
\label{eq:sens}
\end{equation}
the position block stacking
$\mathbf Z\triangleq[\gamma_1\partial\mathbf h_1/\partial\mathbf u_1,
\dots,\gamma_L\partial\mathbf h_L/\partial\mathbf u_L]
\in\mathbb C^{MN\times3L}$ via \eqref{eq:grad} at the true parameters. Eliminating the nuisance block by the Schur complement (the
map $\mathbf X\mapsto\big[\begin{smallmatrix}\Re\mathbf X&-\Im\mathbf X\\
\Im\mathbf X&\Re\mathbf X\end{smallmatrix}\big]$ is a ring homomorphism, so
the $\bm\gamma$-block inverts in closed form) yields
\begin{equation}
\mathbf F_{u|\gamma}=\tfrac{2c}{\sigma^2}\,
\Re\big\{\mathbf Z^{\mathsf H}\mathbf P^{\perp}\mathbf Z\big\},
\qquad
\mathrm{Cov}(\widehat{\mathbf u})\succeq\mathbf F_{u|\gamma}^{-1}.
\label{eq:crlb}
\end{equation}
The projector is the point of profiling: by \eqref{eq:sens}, reflectivity
adjustments move the mean exactly within $\mathrm{span}(\Hc)$, so only the
projected residue $\Zp\triangleq\mathbf P^{\perp}\mathbf Z$ is informative
about $\mathbf u$.

\begin{proposition}\label{prop:ident}
Assume $\Hc$ has full column rank $L$. Then:
(i) \emph{Necessary:} $\mathbf F_{u|\gamma}\succ0$ requires $2MN\ge5L$.
(ii) \emph{Sufficient:} if the complex rank of $\Zp$ equals $3L$, which
requires $MN\ge4L$, then $\mathbf F_{u|\gamma}\succ0$.
(iii) \emph{Exact for $L=1$:} $\mathbf F_{u|\gamma}\succ0$ if and only if the
real $MN\times4$ matrix $\mathbf V$ with rows
$[\mathbf v_{mn}^{\mathsf T},\,1]$, where
$\mathbf v_{mn}\triangleq[(\kb_n-\kb_m)^{\mathsf T},\,k_{m,z}+k_{n,z}]^{\mathsf T}$
collects the per-pair position sensitivities of \eqref{eq:grad}, has rank
$4$. Since every
column of $\mathbf V$ is an additive function of the pair index $(m,n)$,
$\mathrm{rank}(\mathbf V)\le\min(MN,M{+}N{-}1)$: a single target requires
$M{+}N\ge5$ and modes on at least two radii $\|\kb\|$.
For single-snapshot probing ($I=1$), (i) and (ii) hold with $MN\mapsto M$.
\end{proposition}
\begin{IEEEproof}
    Please refer to Appendix~\ref{app:prop2}.
\end{IEEEproof}

The count of (i) is not tight: for $L=1$, part (iii) rules out $MN=3$
despite $2MN\ge5L$, and the split $M=N=2$ despite $MN=4L$, since
$\mathrm{rank}(\mathbf V)\le M{+}N{-}1=3$. Identifiability thus depends on the split and the mode geometry, not on
$MN$ alone, and the rank condition of (ii) must be checked, not assumed
generic, while a genuine gap can survive for $L\ge2$ (numerically,
$(M,N)=(4,2)$, $L=3$ is nonsingular). The coded scheme meets all conditions
with wide margin, whereas single-snapshot probing must meet them with $M$
alone: target capacity is a further dividend of the virtual array.

\subsection{Lattice weight and nested selection}\label{sec:design}
The Fisher matrix has an exact geometric form on the lattice. Group the
pairs by their integer difference and define, with the Kronecker delta, the
lattice weight
\begin{equation}
\omega(\mathbf d)\triangleq\!\sum_{m,n}\delta_{\mathbf d,\,\mathbf t_n-\mathbf r_m}
|c_{mn}(z)|^2,
\label{eq:weight}
\end{equation}
with total weight $\Omega\triangleq\sum_{\mathbf d}\omega(\mathbf d)
=\sum_{mn}|c_{mn}|^2$.

\begin{lemma}\label{lem:fisherlat}
For $L=1$ at $\mathbf u=(\ub,z)$,
\begin{equation}
\mathbf F_{u|\gamma}=\tfrac{2c}{\sigma^2}\,|\gamma|^2\,\Omega\,
\mathrm{Cov}_{\omega}\{\mathbf v\},
\label{eq:fishlat}
\end{equation}
where $\mathrm{Cov}_{\omega}$ denotes covariance under the pair
distribution $|c_{mn}|^2/\Omega$.
\end{lemma}
\begin{IEEEproof}
By \eqref{eq:grad}, the columns of $\mathbf Z$ are
$\jmath\gamma\,\mathbf v_{mn}$ times the entries of $\mathbf h$, and for
$L=1$, $\mathbf P^{\perp}$ projects out $\mathbf h$ itself. Substituting in
\eqref{eq:crlb} and using $|[\mathbf h]_{mn}|^2=|c_{mn}|^2$,
$\Re\{\mathbf Z^{\mathsf H}\mathbf P^{\perp}\mathbf Z\}
=|\gamma|^2\big[\sum_{mn}|c_{mn}|^2\mathbf v\mathbf v^{\mathsf T}
-\tfrac1\Omega\big(\sum_{mn}|c_{mn}|^2\mathbf v\big)
\big(\sum_{mn}|c_{mn}|^2\mathbf v\big)^{\mathsf T}\big]$.
\end{IEEEproof}

\begin{remark}\label{rem:coupling}
For $L\ge2$, $\mathbf F_{u|\gamma}$ is not exactly block diagonal. By
\eqref{eq:hmn}, the normalized cross-Gram
$\mathbf h_\ell^{\mathsf H}\mathbf h_{\ell'}/\Omega$ is the
$\omega$-weighted beampattern of the design evaluated at the separation of
targets $\ell$ and $\ell'$, and substituting $\mathbf P^{\perp}$ in
\eqref{eq:crlb} yields the block diagonal of per-target copies of
\eqref{eq:fishlat} plus diagonal corrections and coupling blocks built
from this beampattern and its first two derivatives at the separations.
Hole-free coverage makes the beampattern a Dirichlet kernel with one-cell
mainlobe, so the couplings fall to sidelobe level once targets are a cell
apart: the spread that maximizes $\mathrm{Cov}_\omega$ also localizes the
couplings. The exact multi-target conditions remain those of
Proposition~\ref{prop:ident}.
\end{remark}

The transverse sensitivities depend on the pair only through the
difference, so the transverse block of \eqref{eq:fishlat} is the
covariance of $\kb_n-\kb_m$ under the normalized lattice weight
$\omega(\mathbf d)/\Omega$: information at a fixed pair budget is the
spread of the lattice weight, not its height. Well inside the propagating disc,
$|c_{mn}|$ is nearly constant and $\omega$ reduces to the pair
multiplicity
$w(\mathbf d)=\sum_{mn}\delta_{\mathbf d,\mathbf t_n-\mathbf r_m}$, the
cross-correlation of the index sets, and since
$\mathbf d=\mathbf t-\mathbf r$ covariances add,
$\mathrm{Cov}_w=\mathrm{Cov}_{\mathcal T}+\mathrm{Cov}_{\mathcal R}$: each
set should be spread as widely as the disc and \eqref{eq:validity} allow.
Coverage also serves the estimator: a whitened acquisition surface (pair weights equalized) is the Fourier transform of $w$, so holes raise sidelobes or ambiguities, while $w\equiv1$ attains the uniform floor.
Maximal hole-free spread at minimal mode count is the nested construction:

\begin{lemma}\label{lem:nested}
Per axis $a$, in units of $2\pi/\mathcal S_a$, pair a dense index set
$\mathcal T=\{-p,\dots,p\}$ with a coarse set
$\mathcal R=(2p{+}1)\cdot\{-s,\dots,s\}$ (stride equal to the dense width).
Then the differences $\{t-r\}$ cover every integer of
$[-(2p{+}1)s-p,\,(2p{+}1)s+p]$ exactly once. Assigning 2-D products of these
sets to the two sides yields a hole-free
$[(2p{+}1)(2s{+}1)]^2$-point difference lattice with $w\equiv1$ from
$(2p{+}1)^2+(2s{+}1)^2$ modes, a count growing only as the square root of the number of covered cells.
\end{lemma}
\begin{IEEEproof}
Divide any $d$ in the range by $2p{+}1$ with centered remainder:
$d=q(2p{+}1)+t$, $|t|\le p$, and $|d|\le(2p{+}1)s+p$ forces $|q|\le s$. Then
$t\in\mathcal T$, $r=-q(2p{+}1)\in\mathcal R$, $d=t-r$, and the
decomposition is unique, whence $w\equiv1$.
\end{IEEEproof}

\begin{remark}\label{rem:resources}
RF chains count simultaneous independent signal ports, not spatial modes. A
waveguide-fed aperture driven by a single feed realizes one programmable
aperture profile per snapshot
\cite{sleasman2017single,shlezinger2021dma}, so the deterministic sequential
code $\mathbf Q=\mathbf I_N$ costs one transmit chain and $I=N$ snapshots,
whereas the $M$ receive projections must be formed simultaneously on each
snapshot, one chain (output port) per mode \cite{shlezinger2021dma}.
\end{remark}

\begin{figure}[!t]
\centering
\includegraphics[width=\columnwidth]{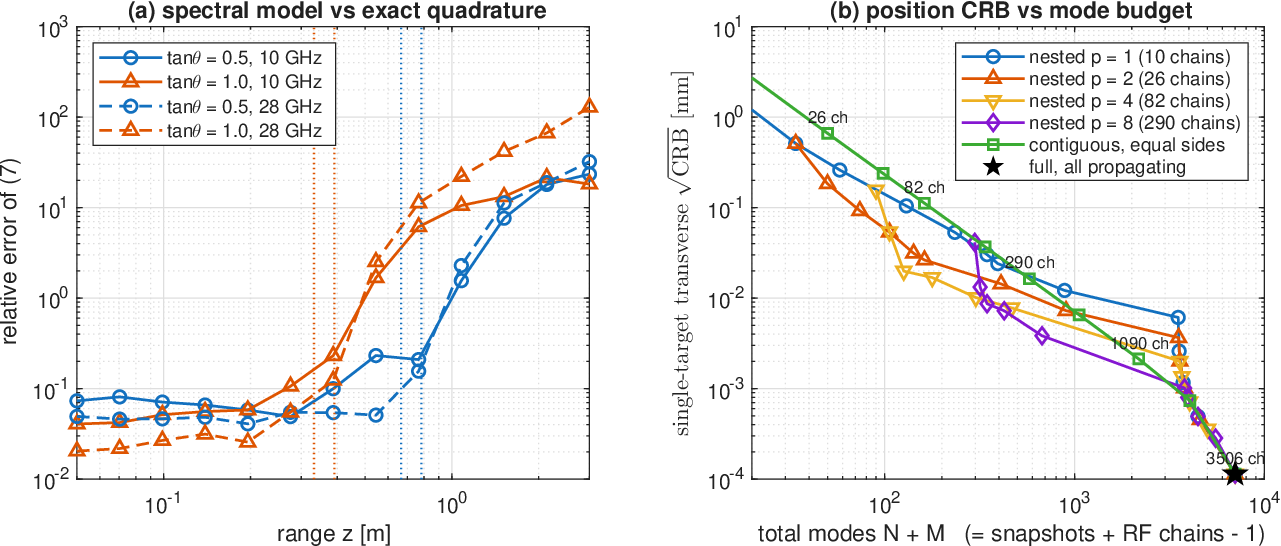}
\caption{(a) Relative error of \eqref{eq:spectral} against exact
quadrature, broadside target, dotted lines marking $z^{\star}$ from \eqref{eq:validity}.
(b) Single-target transverse root CRB against mode budget at fixed noise,
from \eqref{eq:fishlat}.}
\label{fig:model}
\end{figure}

\begin{figure}[!t]
\centering
\includegraphics[width=\columnwidth]{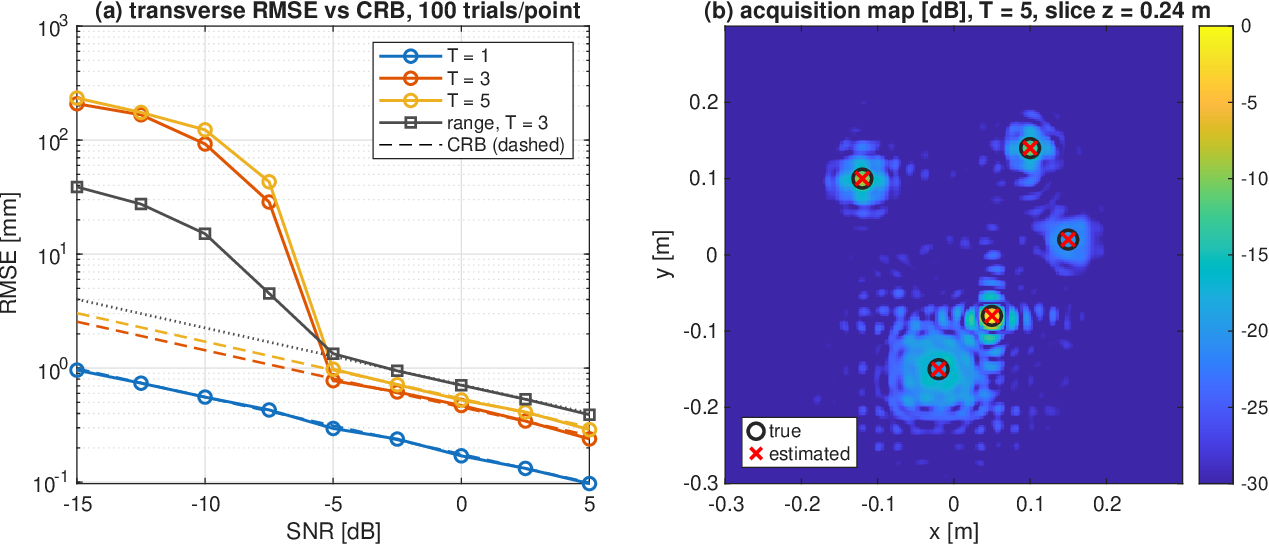}
\caption{Estimation on the $26$-chain design. (a) Transverse RMSE against
the CRB \eqref{eq:crlb}. (b) Whitened acquisition
map with true and estimated positions.}
\label{fig:est}
\end{figure}

\section{Acquisition and Refinement}\label{sec:est}
Given Lemma~\ref{lem:nested}, the estimator exploits \eqref{eq:vecmodel} in three stages.

\emph{Acquisition.} The matched filter
$\mathbf h^{\mathsf H}(\ub,z)\,\mathrm{vec}(\widehat{\mathbf H})
=\sum_{mn}c_{mn}^{*}(z)\widehat H_{mn}e^{-\jmath(\kb_n-\kb_m)\cdot\ub}$
is a 2-D Fourier sum over $\mathcal D$, one FFT per range slice. The matched weights $c_{mn}^{*}(z)$ maximize output SNR, while the
whitened weights $c_{mn}^{*}(z)/|c_{mn}(z)|^{2}$ equalize the
lattice weight \eqref{eq:weight}, removing the rim-heavy $1/(2k_z)$ taper
of \eqref{eq:weyl}. The refinement below is ML in either case. Per axis the largest
difference index is $\bar D=(2p{+}1)s+p$, so the surface is band-limited
with mainlobe half-width $\mathcal S_a/(2\bar D{+}1)$:
an FFT of length $2\bar D{+}1$ with twofold zero-padding lands a grid point within a
quarter cell of every peak, inside the refinement basin. In range, a
single carrier senses $z$ only through
$k_{m,z}+k_{n,z}\in[2\kappa\cos\theta_{\max},2\kappa]$, a coarse scale $\lambda/[2(1-\cos\theta_{\max})]$: a few slices cover the prior and the refinement recovers $z$.

\emph{Refinement.} With $\mathbf Z$ as in Sec.~\ref{sec:crlb}, evaluated at $\bm\gamma=\widehat{\bm\gamma}(\mathbf u)$. With residual
$\mathbf r(\mathbf u)\triangleq\mathbf P^{\perp}
\mathrm{vec}(\widehat{\mathbf H})$, Gauss--Newton with the projected Jacobian
$\mathbf J=-\mathbf P^{\perp}\mathbf Z$ \cite{golub1973varpro} updates
$\mathbf u\leftarrow\mathbf u+\bm\Delta$,
\begin{equation}
\Re\{\mathbf J^{\mathsf H}\mathbf J\}\,\bm\Delta
=-\Re\{\mathbf J^{\mathsf H}\mathbf r\},
\label{eq:gn}
\end{equation}
one projection per step. Started inside the mainlobe it converges in a
few steps, and $\Re\{\mathbf J^{\mathsf H}\mathbf
J\}$ coincides, up to the noise scale, with the Fisher matrix of
Sec.~\ref{sec:crlb}.

\emph{Multiple targets.} CLEAN-type deflation with joint refits: refine the strongest peak via \eqref{eq:gn}, subtract
$\Hc\widehat{\bm\gamma}$ from $\mathrm{vec}(\widehat{\mathbf H})$, reapply
the filter to the residual, and finish with a joint refinement over all
$3L$ coordinates.

\section{Numerical Results}\label{sec:num}
All experiments use a square aperture of side $\mathcal S_x=\mathcal S_y=1$\,m at $f_0=10$\,GHz
($\lambda=3$\,cm, $33\lambda$ per side), and a common tangential polarization
$\mathbf p_n=\mathbf w_m=\hat{\mathbf x}$. The
reference design is the nested selection of Lemma~\ref{lem:nested} with
$(p,s)=(2,5)$: $5^2{=}25$ dense receive modes acquired simultaneously and
$11^2{=}121$ coarse transmit modes probed sequentially
($\mathbf Q=\mathbf I_N$, $c=1$), of which four evanescent corner modes are
discarded, giving $I=N=117$ snapshots, $26$ RF chains, and a $55\times55$
difference lattice with transverse cell
$\mathcal S_x/55\approx18$\,mm, hole-free up to the four
corner blocks removed because their modes are evanescent ($2925$ cells).
Acquisition uses whitened weights, twofold zero padding, range slices spaced $20$\,mm over $[0.18,0.32]$\,m, refinement \eqref{eq:gn}, and SNR is per decoded entry.

Fig.~\ref{fig:model}(a) probes Proposition~\ref{prop:valid}: the relative
error of \eqref{eq:spectral} against exact quadrature of \eqref{eq:alphab}
for a broadside target as a function of range, for walk-off slopes
$\tan\theta_n\in\{0.5,1\}$ at two carriers. The error ramps around the ranges $z^{\star}(\theta_n,f_0)$ solving \eqref{eq:validity} with equality (a leading-order marker). The $28$\,GHz knees sit
only slightly beyond the $10$\,GHz ones, since the walk-off term of
\eqref{eq:validity} is purely geometric and only the Fresnel radius
shrinks with the carrier, while the few-percent floor is the edge-diffraction residue of order $(\kappa z)^{-1/2}$ and drops with the carrier accordingly. Beyond $z^{\star}$ the error exceeds unity,
an amplitude divergence, which is why every mode selection must respect
the cone of Remark~\ref{rem:near}.

Fig.~\ref{fig:model}(b) turns Lemma~\ref{lem:fisherlat} into a design
chart: the single-target transverse root CRB at $z=0.25$\,m and a noise level fixed so that the reference design operates at $15$\,dB SNR per decoded entry, evaluated through
\eqref{eq:fishlat} directly on the index sets. The abscissa counts modes, is multiplexing-invariant by Lemma~\ref{lem:code}, and under sequential probing equals snapshots plus chains minus one. Each nested curve fixes its receive hardware, $p\in\{1,2,4,8\}$ ($10$ to $290$ chains), and grows in three regimes, coarse factor $s$, transmit stride refinement (both paid in snapshots alone), then receive growth, terminating at the all-propagating system (star), while the contiguous family grows both sides throughout. The chart quantifies the trade of Sec.~\ref{sec:design}. At equal budget the nested designs lie below, $0.026$\,mm at ${\approx}160$ modes and $26$ chains against $0.112$\,mm at $82$ chains, a factor $4.3$: spreading the pairs maximizes $\mathrm{Cov}_\omega$, while multiplicity re-measures the same differences. Prolonged by snapshots alone, the $26$-chain curve reaches
${\approx}4\,\mu$m with the transmit side occupying the whole disc, within
a small factor of the full system: Remark~\ref{rem:resources} in a single
point. Multiplicity buys the taper, a $-26$\,dB floor for the triangular contiguous $w$ against $-13$\,dB for $w\equiv1$ (Lemma~\ref{lem:nested}): a soft advantage, since sidelobes below the dynamic range are removed by deflation, whereas an unexplored difference no processing can recover.

Fig.~\ref{fig:est} validates the estimator on the reference design for
$L\in\{1,3,5\}$ targets drawn from
$(0.05,-0.08,0.24)$, $(-0.12,0.10,0.22)$, $(0.15,0.02,0.28)$,
$(-0.02,-0.15,0.30)$, $(0.10,0.14,0.26)$\,m, reflectivity magnitudes spanning $10$\,dB with random phases, $100$ trials per SNR point, everything else fixed. Above an $L$-dependent threshold
the transverse RMSE attains the bound \eqref{eq:crlb} within the sampling error of the finite trial count, while below it acquisition
outliers dominate and the threshold migrates right with $L$ as deflation
must separate progressively weaker returns. Slight excursions below the bound are expected, from the finite trial count and from the small finite-SNR bias of the ML estimator, to which the unbiased bound does not apply. The persistent gap
between range and transverse accuracy (shown for $L=3$) is the
single-carrier imbalance of Sec.~\ref{sec:est}, range information entering
only through $k_{m,z}+k_{n,z}$. In the $L=5$ realization of Fig.~\ref{fig:est}(b), all targets are acquired within one cell despite the $-13$\,dB floor: deflation absorbs the sidelobe cost of $w\equiv1$.

\section{Conclusion}
Single-snapshot holographic probing of point targets is rank-limited to $3L$.
Invertible transmit coding lifts the limit and places the decoded data on a
wavenumber-domain difference lattice. The specular-point condition bounds
the modal support of the spectral model, and the Fisher matrix reduces to a
covariance over the lattice, turning mode selection into a variance
maximization whose hole-free optimum is the nested design: full-aperture resolution from tens of RF chains, since transmit modes cost time rather than hardware. The identifiability analysis quantifies the target capacity gained. Wideband operation, addressing the
range and cross-range information imbalance of single-carrier sensing, is
the natural extension.

\appendices
\section{Proof of Proposition~\ref{prop:valid}}\label{app:prop1}
\begin{IEEEproof}
Write $\mathbf G(\mathbf u_\ell,\mathbf s)=\mathbf D(\mathbf u_\ell-\mathbf s)\,
e^{\jmath\kappa R(\sbb)}$, $R(\sbb)=\sqrt{z_\ell^2+\|\sbb-\ub_\ell\|^2}$. The integrand is a slow amplitude times
$e^{\jmath\Phi(\sbb)}$, $\Phi=\kappa R+\kb_n\cdot\sbb$, oscillating on
the scale $\lambda\ll R$, so the integral is dominated by the stationary
set of $\Phi$. Stationarity requires $\kappa(\sbb-\ub_\ell)/R=-\kb_n$,
whence $R=z_\ell/\cos\theta_n$ and the unique point
$\sbb^{\star}_{n,\ell}=\ub_\ell-z_\ell\kb_n/k_{n,z}$, from which a ray along
$+\mathbf k_n$ reaches the target. Around it the linear term vanishes and
$\Phi(\sbb^{\star}{+}\bm\nu)=\Phi(\sbb^{\star})
+\tfrac12\bm\nu^{\mathsf T}\nabla^2\Phi\,\bm\nu+\dots$, with
$\nabla^2\Phi=\kappa\big[\mathbf I_2/R-(\sbb-\ub_\ell)(\sbb-\ub_\ell)^{\mathsf T}
/R^3\big]$ of eigenvalues $\kappa\cos^3\theta_n/z_\ell$ and
$\kappa\cos\theta_n/z_\ell$. Points with quadratic phase below $\pi$ add coherently, beyond them
successive Fresnel annuli cancel pairwise. Setting
$\tfrac12(\kappa\cos^3\theta_n/z_\ell)r^2=\pi$ along the weaker axis gives the
first-zone radius $r_F=\sqrt{\lambda z_\ell/\cos^3\theta_n}$, and the plane integral is dominated by this disc. Since
$\int_{\mathcal S}=\int_{\mathbb R^2}-\int_{\mathbb R^2\setminus\mathcal
S}$, if \eqref{eq:validity} holds the removed region contains no
stationary point, integration by parts along $\nabla\Phi$ suppresses the
oscillatory remainder, and the surviving edge-diffraction term is smaller
than the disc contribution by $\mathcal O((\kappa z_\ell)^{-1/2})$. If it
fails, the truncation removes an $\mathcal O(1)$ fraction of the dominant
term.
\end{IEEEproof}

\section{Proof of Proposition~\ref{prop:ident}}\label{app:prop2}
\begin{IEEEproof}
Since $\mathbf x^{\mathsf T}\mathbf F_{u|\gamma}\mathbf x
=\tfrac{2c}{\sigma^2}\|\Zp\mathbf x\|^2$, singularity is equivalent to the
existence of a real $\mathbf x\neq\mathbf 0$ and $\mathbf c\in\mathbb C^L$
with $\mathbf Z\mathbf x=\Hc\mathbf c$, in which case
$\bm\mu(\mathbf u{+}\epsilon\mathbf x,\bm\gamma{-}\epsilon\mathbf c)
=\bm\mu+\mathcal O(\epsilon^2)$ by \eqref{eq:sens}: displaced targets are
masked by re-fitted reflectivities. Note that $\mathbf x$ is real while
$\mathbf c$ is complex.

(i) In $\mathbb R^{2MN}$, $\{\mathbf Z\mathbf x\}$ spans a subspace
of dimension $\le3L$ and $\{\Hc\mathbf c\}$ one of dimension $2L$
(directions $\mathbf h_\ell$, $\jmath\mathbf h_\ell$). If $5L>2MN$ they
intersect nontrivially and a masking pair exists.

(ii) If $\mathrm{rank}_{\mathbb C}(\Zp)=3L$, the columns of $\Zp$ are
independent over $\mathbb C$, hence over $\mathbb R$, so
$\Zp\mathbf x\neq\mathbf 0$ for every real $\mathbf x\neq\mathbf 0$. Since
$\mathbf P^{\perp}$ projects onto a subspace of complex dimension $MN-L$,
rank $3L$ requires $MN\ge4L$.

(iii) For $L=1$, \eqref{eq:grad} gives $[\mathbf Z\mathbf x]_{(m,n)}
=\gamma\jmath\big[(\kb_n-\kb_m)\cdot\bar{\mathbf x}
+(k_{m,z}+k_{n,z})x_3\big][\mathbf h]_{(m,n)}$ and
$[\Hc c]_{(m,n)}=c\,[\mathbf h]_{(m,n)}$, with all entries of $\mathbf h$
nonzero, so $\mathbf Z\mathbf x=\Hc c$ iff
$(\kb_n-\kb_m)\cdot\bar{\mathbf x}+(k_{m,z}+k_{n,z})x_3=\mu$ for every
pair, $\mu\triangleq c/(\jmath\gamma)$ forced real, i.e., iff
$\mathbf V[\mathbf x^{\mathsf T},-\mu]^{\mathsf T}=\mathbf 0$, where any
nontrivial null vector has $\mathbf x\neq\mathbf 0$: singularity is
$\mathrm{rank}(\mathbf V)<4$. Each column of $\mathbf V$ is additive, and such functions span a space of dimension $M{+}N{-}1$,
capping the rank. A common radius makes $k_{m,z}+k_{n,z}$ constant and
collapses the third column onto the fourth.
\end{IEEEproof}

\bibliographystyle{IEEEtran}
\bibliography{refs}

\end{document}

%% file: acronyms.tex
\acrodef{CRB}{Cram\'er--Rao bound}
\acrodef{RF}{radio frequency}
\acrodef{SNR}{signal-to-noise ratio}
\acrodef{MIMO}{multiple-input multiple-output}
\acrodef{ML}{maximum likelihood}
\acrodef{FFT}{fast Fourier transform}
\acrodef{DFT}{discrete Fourier transform}
\acrodef{CAPA}{continuous-aperture array}